\documentclass[11pt]{llncs}
\usepackage{amsmath,amssymb,amsbsy,amsfonts,latexsym,
               amsopn,amstext,amsxtra,euscript,amscd}

\usepackage{color}

\def\F{{\mathbb F}}

\def\F{{\mathbb F}}

\def\ll{{\bf \lambda}}
\def\00{{\bf 0}}
\def\+{\oplus}

\def\\{\cr}
\def\({\left(}
\def\){\right)}

\providecommand{\newoperator}[3]{%
  \newcommand*{#1}{\mathop{#2}#3}}

\newoperator{\FD}{\mathrm{FD}}{\nolimits}

\newcommand{\GF}[2][2]{{\mathbb F}_{#1^{#2}}}

\numberwithin{thm}{section}

\newcommand{\tr}[2][1]{Tr_{#1}^{#2}}
\def\urltilda{\kern -.15em\lower .7ex\hbox{\~{}}\kern .04em}

\begin{document}
\title{Further study on the maximum number of bent components of
vectorial functions}

\author{ Sihem Mesnager$^{1}$, Fengrong Zhang$^2$, Chunming Tang$^3$, Yong Zhou$^2$}
\institute{
1. LAGA, Department of Mathematics, University of
Paris VIII \\(and Paris XIII and CNRS),
 Saint--Denis cedex 02, France.\\
E-mail: \email{smesnager@univ-paris8.fr}\\
2. School of Computer Science and Technology,
 China University\\ of Mining and Technology,
 Xuzhou, Jiangsu 221116, China.\\
E-mail: \email{\{zhfl203,yzhou\}@cumt.edu.cn}\\
3. School of Mathematics and Information, China West Normal University, Nanchong, Sichuan 637002, China.\\
E-mail: \email{tangchunmingmath@163.com}
}



\date{\today}
\maketitle

\begin{abstract}
In 2018, Pott, at al. have studied in [IEEE Transactions on Information Theory. Volume: 64, Issue: 1, 2018] the  maximum number of bent components of vectorial function. They have presented serval
nice results  and suggested several open problems in this context. This paper is in the continuation of their study in which we solve two open problems raised by Pott et al. and partially solve an open problem raised by the same authors.
Firstly, we prove that  for a vectorial  function, the property of having the maximum number of bent components is invariant under the so-called CCZ equivalence. Secondly, we prove the non-existence of APN plateaued having the maximum number of bent components. In particular, quadratic
APN functions cannot have the maximum number of bent components.
Finally, we  present some sufficient conditions that the vectorial function defined from $\mathbb{F}_{2^{2k}}$ to $\mathbb{F}_{2^{2k}}$  by its univariate representation:
  $$ \alpha x^{2^i}\left(x+x^{2^k}+\sum\limits_{j=1}^{\rho}\gamma^{(j)}x^{2^{t_j}}
  +\sum\limits_{j=1}^{\rho}\gamma^{(j)}x^{2^{t_j+k}}\right)$$ has the maximum number of {  components bent functions, where $\rho\leq k$}. Further, we show that the differential spectrum of the
   function $   x^{2^i}(x+x^{2^k}+x^{2^{t_1}}+x^{2^{t_1+k}}+x^{2^{t_2}}+x^{2^{t_2+k}})$ (where $i,t_1,t_2$ satisfy some conditions) is different from the binomial function $F^i(x)= x^{2^i}(x+x^{2^k})$ presented in the article of Pott et al.

  Finally, we provide  sufficient and necessary conditions so that
  the functions $$Tr_1^{2k}\left(\alpha x^{2^i}\left(Tr^{2k}_{e}(x)+\sum\limits_{j=1}^{\rho}\gamma^{(j)}(Tr^{2k}_{e}(x))^{2^j}
  \right)\right) $$ are bent.

\end{abstract}
{\bf Keywords:} Vectorial functions, Boolean functions, Bent functions, Nonlinearity, APN functions, Plateaued functions, CCZ equivalence. \bigskip
\section{Introduction}

Vectorial (multi-output) Boolean functions, that is, functions from the vector space $\mathbb{F}^{n}_{2}$ (of all binary
vectors of length $n$) to the vector space $\mathbb{F}^{m}_{2}$, for given positive integers
$n$ and $m$. These functions are called $(n,m)$-functions and
include the (single-output) Boolean functions (which correspond to the case $m
=1$). In symmetric cryptography,  multi-output functions are called \emph{S-boxes}.
They are fundamental parts of block ciphers.
Being the only source of nonlinearity in these ciphers, S-boxes play a
central role in their robustness, by providing confusion (a requirement already mentioned by C. Shannon), which is necessary to withstand known (and hopefully future) attacks.
 When they are used as S-boxes in block ciphers, their number $m$ of output bits equals or approximately equals the number $n$ of input bits. They can also be used in stream ciphers, with $m$ significantly smaller than $n$, in the place of Boolean functions to speed up the ciphers.

We shall identify $\mathbb{F}_{2}^n$ with the Galois field $\mathbb{F}_{2^n}$ of order $2^n$ but we shall always use $\mathbb{F}_{2}^n$ when the field structure will not really be used.
The {\em component functions} of $F$ are the Boolean functions $v\cdot F$, that is, $x\in\mathbb{F}_{2^n}\mapsto  \tr{m}(v F(x))$, where ``$\cdot$" stands for an inner product in $\mathbb{F}_{2^m}$ (for instance:
 $u\cdot v:=\tr m(uv), \forall u\in \mathbb{F}_{2^m}, v\in \mathbb{F}_{2^m}$ where "$\tr m$" denotes the absolute trace over $ \mathbb{F}_{2^m}$).  In order to classify vectorial Boolean functions that satisfy desirable
nonlinearity conditions, or to determine whether, once found, they
are essentially new (that is, inequivalent in some sense to any of
the functions already found) we use some concepts of equivalence.
For vectorial Boolean functions, there exist essentially two kinds concepts of equivalence:
 the extended affine EA-equivalence and the
CCZ-equivalence (Carlet-Charpin-Zinoviev equivalence). Two $(n,r)$-functions $F$ and
$F^{\prime}$ are said to be EA-equivalent if there exist affine
automorphisms $L$ from $\mathbb{F}_{2^n}$ to $\mathbb{F}_{2^n}$ and $L^{\prime}$ from $\mathbb{F}_{2^r}$
to $\mathbb{F}_{2^r}$ and an affine function $L''$ from $\mathbb{F}_{2^n}$ to $\mathbb{F}_{2^r}$ such
that $F'=L^{\prime} \circ F \circ L + L''$. EA-equivalence is a
particular case of CCZ-equivalence \cite{CCZ98}. Two
$(n,r)$-functions $F$ and $F^{\prime}$ are said to be CCZ-equivalent if
their
graphs $G_F:=\{(x,F(x)),~ x\in \mathbb{F}_{2^n}$\} and
$G_F^{\prime}:=\{(x,F^{\prime}(x)),~ x\in \mathbb{F}_{2^n}$\}
are affine equivalent, that is, if there exists an affine
permutation $\mathcal{L}$ of $\mathbb{F}_{2^n} \times \mathbb{F}_{2^m}$ such that
$\mathcal{L}(G_F)=G_F^{\prime}$.

A standard notion of \emph{nonlinearity}
of an $(n,m)$-function $F$ is defined as
\begin{equation}\label{N_1}
\mathcal{N}(F)=\underset{v\in
\F_{2^m}^\star}{\min}nl(v\cdot F),
\end{equation}
where $v\cdot F$
denotes the usual inner product on $\F_{2^m}$ and $nl(\cdot)$ denotes the nonlinearity of Boolean functions (see definition in Section \ref{Preliminaries}).
From the covering radius bound, it is known that $\mathcal{N}(F)\leqslant 2^{n-1}-2^{n/2-1}$. The functions achieving this bound are called $(n,m)$-\emph{bent} functions. Equivalently, a vectorial Boolean function $F: \mathbb{F}_{2^n}\rightarrow \mathbb{F}_{2^m}$ is said to be a vectorial bent function if  all nonzero component functions of $F$ are \emph{bent} (Boolean) functions. Bent Boolean functions have  maximum  Hamming distance to the set of affine Boolean functions. The notion of bent function was introduced by Rothaus \cite{RO76} and attracted a lot of research of more than four decades. Such functions
are extremal combinatorial objects with several areas of application, such as coding theory, maximum length sequences, cryptography. A survey on bent function can be found in \cite{CarletMesnagerDCC2016} as well as the book \cite{MesnagerBook}.

In~\cite{CC-Nyberg}, it is shown that $(n,m)$-bent functions exist only if $n$ is even and $m\leqslant n/2$.
The notion of nonlinearity in (\ref{N_1}) (denoted by $\mathcal{N}$), was first introduced by Nyberg in~\cite{CC-Nyberg}, which is closely related to Matsui's linear attack~\cite{Matsui} on block ciphers. It has been  further studied
by Chabaud and Vaudenay \cite{ChabaudVaudenay95}. The nonlinearity is
invariant under CCZ equivalence (and hence under extended affine equivalence).
 Budaghyan and  Carlet have proved in \cite{BC09} that for
bent vectorial Boolean functions, CCZ-equivalence coincides with
EA-equivalence.


The problem of construction vectorial bent functions has been considered in the literature. Nyberg \cite{CC-Nyberg} investigated the constructions of vectorial bent functions; she presented two constructions based on Maiorana-McFarland bent functions and $\mathcal{PS}$ bent functions, respectively. In \cite{SatohIwataKurosawa99},  Satoh, Iwata, and  Kurosawa have improved the first method of construction given in \cite{CC-Nyberg} so that the resulting functions achieve the largest degree. Further, serval constructions of bent vectorial functions have been investigated in some papers \cite{CarletMesnagerBentVectorial,Pasaliczh2012,Feng2011,MesnagerDCC2015,Wu2005}. A complete state of the art can be found in \cite{MesnagerBook} (Chapter 12).

 Very recently, Pott \emph{et al}.\cite{Pott2017} considered functions $ \mathbb{F}_{2^n}\rightarrow \mathbb{F}_{2^n}$ of the form
$F^i(x) = x^{2^i}(x + x^{2^{k}})$, where $n=2k, i=0,1,\cdots,n-1$. They showed that the upper bound of number
of bent component functions of a vectorial function  $F: \mathbb{F}_{2^n}\rightarrow \mathbb{F}_{2^n}$ is $2^n-2^{n/2}$ ($n$ even). In addition, they showed that the
binomials $F^i(x)= x^{2^i}(x + x^{2^k})$  have such a large number
of bent components, and these binomials are inequivalent to the
monomials $x^{2^{k}+1}$ if $0 < i < k$. Further, the properties  (such as differential properties and complete Walsh spectrum)  of the functions $F^i$ were investigated.

  In this paper, we will consider three open problems raised by Pott et al \cite{Pott2017}. In the first part,  we prove that CCZ equivalence is preserved for vectorial  functions having the maximum number of bent components. Next, we consider APN plateaued functions and investigate if they can have the maximum number of bent components. We shall give a negative answer to this question. Finally, we consider the bentness property of
  functions  $\mathbb{F}_{2^{2k}}\rightarrow \mathbb{F}_{2^{2k}}$ of the form
  \begin{equation}\label{equa main}
  G(x)=\alpha x^{2^i}\left(x+x^{2^k}+\sum\limits_{j=1}^{\rho}\gamma^{(j)}x^{2^{t_j}}
  +\sum\limits_{j=1}^{\rho}\gamma^{(j)}x^{2^{t_j+k}}\right),
  \end{equation}
where
 $m\leq k$, $\gamma^{(j)}\in {\Bbb F}_{2^{k}}$ and $0\leq t_j\leq k$ be a nonnegative
integer.
In particular,  { we show the functions   $ x^{2^{t_2}}\left(x+x^{2^k}+x^{2^{t_1}}+x^{2^{t_1+k}}+x^{2^{t_2}}+x^{2^{t_2+k}}\right)$ are  inequivalent to $x^{2^{t_2}}(x + x^{2^k}) $,  where $t_1=1$ and $\gcd(t_2,k) \neq 1$ }.
Here we use the concept
of CCZ-equivalence when we speak about the equivalence
of functions.
The rest of the paper is organized as follows. Some preliminaries are given in Section \ref{Preliminaries}. In Section \ref{stability}, we prove
our result on the stability under CCZ equivalence of a function having the maximum number of bent components which solve Problem 4 in
\cite{Pott2017}. Next, in Section \ref{APNplateaued}, we prove that APN plateaued functions cannot have the maximum number of bent components, which partially solves Problem 8 in \cite{Pott2017}. Finally, in Section \ref{newbent} we investigate Problem 2 in \cite{Pott2017}. To this end,  we provide several functions defined as $G(x) = x  L(x)$ on ${\Bbb F}_{2^{2k}}$ (where $L(x)$ is a linear function on  ${\Bbb F}_{2^{2k}}$)
 such that the number of bent components $Tr^{2k}_1(\alpha F(x))$ is maximal.

\section{Preliminaries and notation}\label{Preliminaries}
 Throughout this article, $\| E\|$ denotes the cardinality of a finite set $E$, the binary field is denoted by $\mathbb{F}_2$ and the finite field of order $2^n$ is denoted by ${\Bbb F}_{2^n}$. The multiplicative group ${\Bbb F}^*_{2^n}$ is a cyclic group consisting
of $2^n-1$ elements.  The set of all Boolean functions mapping from ${\Bbb F}_{2^n}$ (or $ {\Bbb F}_2^n$) to ${\Bbb F}_{2}$ is denoted by $B_n$.

 Recall that for any positive integers $k$, and $r$
dividing $k$, the trace function from $\GF{k}$ to $\GF{r}$, denoted
by $Tr_{r}^{k}$, is the mapping defined as:
\begin{displaymath}
  Tr_{r}^{k}(x):=\sum_{i=0}^{\frac kr-1}
  x^{2^{ir}}=x+x^{2^r}+x^{2^{2r}}+\cdots+x^{2^{k-r}}.
\end{displaymath}
   In particular, the {\em absolute trace} over $\mathbb{F}_2$ of an element $x \in
\mathbb{F}_{2^n}$ equals $Tr_1^{n}(x)=\sum_{i=0}^{n-1} x^{2^i}$.

   There exist several kinds of possible {\em univariate representations}  (also called trace, or polynomial, representations) of Boolean functions  which are not all unique and use the identification between the vector-space ${\Bbb F}_2^n$ and the field $\GF{n}$.  Any Boolean function over ${\Bbb F}_2^n$ can be represented in a unique way as a polynomial in one variable $x\in \mathbb{F}_{2^n}$ of the form $f(x) =\sum_{j=0}^{2^n-1} a_j x^j$, where $a_0,a_{2^n-1}\in{\Bbb F}_2$, $a_j$'s are elements of $\GF{n}$ for $1\leq j <2^n-1$
 such that ${a_j}^2=a_{2i\mod (2^n-1)}$. The binary expansion of $j$ is $j=j_0+j_12+\cdots j_{n-1}2^{n-1}$ and we denote $\bar{j}=(j_0,j_1, \cdots,j_{n-1})$. The algebraic degree of $f$ equals $max \{wt(\bar i)\mid a_j\not=0, 0\leq j < 2^n\}$ where $wt(\bar i)=j_0+j_1+\cdots+j_{n-1}$. Affine functions (whose set is denoted by $A_n$) are those of algebraic degree at most
$1$. The Walsh transform of
$f\in{ B}_{n}$ at $ \ll \in \mathbb{F}_{2^n}$ is defined as
$$ \begin{array}{l}W_f(\ll) = \sum\limits_{x \in \mathbb{F}_{2^n}}(-1)^{f(x)
+ Tr^n_1(\ll x)}.\end{array}$$

  The \emph{nonlinearity} of $f\in B_{n}$ is defined as the minimum Hamming distance
  to the set of all $n$-variable affine functions, i.e.,
\begin{displaymath}
\begin{array}{c}
     nl(f)=\min_{g\in A_{n}}d(f,g).
\end{array}
  \end{displaymath}
  where $d(f,g)$ is the Hamming distance between $f$ and $g$.
Following is the relationship
between nonlinearity and Walsh spectrum of $f \in {B}_n$ $$
 nl(f) = 2^{n-1} - \frac{1}{2}\max_{\ll \in
 \mathbb{F}_{2^n}}|W_f(\ll)|.$$
 By Parseval's identity $\sum_{\ll \in
 \mathbb{F}_{2^n}} W_f(\ll)^2 = 2^{2n}$, it can be shown that
 \noindent $ \max\{ |W_f(\ll) | : \ll \in \mathbb{F}_{2^n} \} \ge
 2^{\frac{n}{2}}$ which implies that $nl(f) \le 2^{n-1} - 2^{\frac{n}{2} -
 1}$.
 If $n$ is an even integer a function $f \in \mathcal{B}_n$ is
said to be \emph{ bent }  if  $W_f(\ll) \in \{2^{\frac{n}{2}},
-2^{\frac{n}{2}}\}$, for all $\ll \in \mathbb{F}_{2^n}$. Moreover, a function $f \in \mathcal{B}_n$ is said to be $t$-\emph{plateaued} if $W_f(\ll) \in\{ 0, \pm 2^\frac{n+t}{2}\}$, for all $\ll \in \mathbb{F}_{2^n}$. The integer $t$ $(0\leq t\leq n$) is called the \emph {amplitude} of $f$. Note that  a bent function is a 0-plateaued function.
In the following, $``<,>"$ denotes the standard inner (dot) product of two vectors, that is, $<\ll , x>=\lambda_1x_1+ \ldots + \lambda_nx_n$, where $\ll,x\in F_2^n$.  If we  identify the vector space $F_2^n$ with
the finite field $F_{2^n}$, we use the trace bilinear form
$Tr^n_1(\ll x)$  instead of the dot product,  that is, $<\ll , x>=Tr^n_1(\ll x)$, where $\ll,x\in F_{2^n}$.

For vectorial functions $F:  \mathbb{F}_2^n\rightarrow \mathbb{F}_2^m$, the extended Walsh-Hadamard transform defined as,
$$ \begin{array}{l}
 W_F(u,v) = \sum\limits_{x \in \mathbb{F}_2^n}(-1)^{<v,F(x)>
+ <u, x>},\end{array}$$
where  $F(x)=(f_1(x),f_2(x),\cdots,f_m(x)),u\in \mathbb{F}_2^n, v\in \mathbb{F}_2^m$.

Let $F$ be a vectorial function from $\mathbb{F}_{2^n}$ into $\mathbb{F}_{2^m}$. The linear combinaison of the coordinates of $F$ are the Boolean functions $f_\lambda: x\mapsto Tr_1^{m} (\lambda F(x))$, $\lambda \in \mathbb{F}_{2^m}$, where $f_{0}$ is the null function. The functions  $f_\lambda$ ($\lambda\not=0$) are called the \emph{components} of $F$. A vectorial function is said to be \emph{bent} (resp.  $t$-\emph{plateaued}) if all its components are bent (resp. $t$-plateaued). A vectorial function $F$  is called \textit{vectorial plateaued} if all its components are plateaued with possibly different amplitudes.


  Let $F:  \mathbb{F}_2^n\rightarrow \mathbb{F}_2^n $ be an $(n,n)$-function. For any $a\in {\Bbb F}_2^n, b\in {\Bbb F}_2^n$,
  we denote
  \begin{displaymath}
  \begin{array}{c}
  \Delta_F(a,b)=\{x|x\in {\Bbb F}_2^n,F(x\oplus a)\oplus F(x)=b\},\\
  \delta_F(a,b)=\|\Delta_F(a,b)\|,
  \end{array}
  \end{displaymath}
  Then, we have $\delta(F):=max_{a\not=0, b\in {\Bbb F}_2^n }  \delta_F(a,b)\geq 2$ and the functions for which equality holds are said to be \emph {almost perfect nonlinear}(APN).


 A  nice survey on Boolean and vectorial Boolean functions for cryptography can be found in \cite{Cbook1} and \cite{Cbook}, respectively.

\section{The stability of a function having the maximum number of bent components under CCZ equivalence}\label{stability}
In \cite{Pott2017}, Pott et al. have shown that the maximum number of bent components of a vectorial $(n,n)$-function $F$ is $2^n -2^k$ where $k:=\frac{n}2$ ($n$ even). They left open the problem whether the property of a function having the maximum number of bent components is invariant under CCZ equivalence or not. In this section we solve this problem by giving a positive answer in the following theorem.

 \begin{theorem}\label{theo:CCZ}
 Let $n=2k$ and $F, F'': \mathbb F_2^n \rightarrow \mathbb F_2^n$  be CCZ-equivalent functions. Then $F$ has $2^n -2^k$ bent components if and only if
$F''$ has $2^n -2^k$ bent components.
  \end{theorem}
 \begin{proof}
 Let $F$ be a function with $2^n -2^k$ bent components.
 Define
 \begin{displaymath}
  S=\{v\in \mathbb F_2^n: x \rightarrow <v, F(x)> \text{ is not bent} \}.
 \end{displaymath}
 By Theorem 3 of \cite{Pott2017}, $S$  is a linear
subspace of dimension $k$. Then, let $U$ be any $k$-dimensional subspace of $\mathbb F_2^n$ such that $U \cap S =\{0\}$.
 Let $v_1, \cdots, v_k$ be a basis of $S$ and $u_1, \cdots, u_k$ be a basis of $U$.
 Define a new function $F': \mathbb F_2^n \rightarrow \mathbb F_2^n$ as
 \begin{displaymath}
  F'(x)=(H(x), I(x))
 \end{displaymath}
 where $H(x)=(<v_1, F(x)>, \cdots, <v_k, F(x)>)$
  and $I(x)=(<u_1, F(x)>, \cdots, <u_k, F(x)>)$.
  Then, $F'$ is EA-equivalent to $F$. Recall that the property of a function having the maximum number of bent components is invariant under EA equivalence.
  Thus, $F'$ has $2^n -2^k$ bent components. Since $F$ and $F''$  are CCZ-equivalent functions, $F''$ is CCZ-equivalent to $F'$, which has $2^n -2^k$ bent components.
  Let $\mathcal L(x, y,z) = (L_1(x, y,z), L_2(x, y,z), L_3(x, y,z))$, (with $L_1 :\mathbb F^n_2 \times F^k_2 \times F^k_2 \rightarrow \mathbb F^n_2$,
  $L_2 :\mathbb F^n_2 \times F^k_2 \times F^k_2 \rightarrow \mathbb F^k_2$ and  $L_3 :\mathbb F^n_2 \times F^k_2 \times F^k_2 \rightarrow \mathbb F^k_2$)
  be an affine permutation of $\mathbb F^n_2 \times F^k_2 \times F^k_2$ which maps the graph of $F'$ to the graph of $F''$.
  Then, the graph $\mathcal G_{F''}=\{\mathcal L (x, H(x), I(x)): x \in \mathbb F_2^n \}$. Thus $L_1(x,H(x), I(x))$ is a permutation
  and for some affine function $L_1': \mathbb F^n_2 \times F^k_2 \rightarrow \mathbb F^n_2$ and linear function $L_1'':  \mathbb F^k_2 \rightarrow \mathbb F^n_2$ we can write
  $L_1(x, y,z)=L_1'(x, y)+L_1''(z)$.
  For any element $v$ of $\mathbb F^n_2$ we have
  \begin{align}\label{eq:v-comp}
 <v , L_1(x, H (x), I(x))> =& <v , L_1'(x,H(x))> + <v , L_1''(I (x))> \nonumber\\
  =& <v , L_1'(x,H(x))> + <L_1''^*(v) , I (x)> \nonumber\\
  =& <L_1''^*(v) , I (x)>+<v',H(x)>+<v'',x>+a,
  \end{align}
where $a\in \mathbb F_2$, $v'\in \mathbb F_2^k$, $v''\in \mathbb F_2^n$ and $L_1''^*$  is the adjoint operator of $L_1''$, in fact, $L_1''^*$ is the
linear permutation whose matrix is transposed of that of $L_1''$.
 Since $L_1(x,H(x), I(x))$ is a permutation, then any function $<v , L_1(x, H (x), I(x))>$ is balanced (recall that this property is a necessary and sufficient condition) and,
  hence, cannot be bent.
  From the construction of $F'$, $<L_1''^*(v) , I (x)>+<v',H(x)>+<v'',x>+a$ is not bent if and only if $L_1''^*(v)=0$. Therefore,
  $L_1''^*(v)=0$ for any $v\in \mathbb F_2^n$. This means that $L_1''$ is null, that is, $L_1(x, H (x), I(x))=L_1'(x, H (x)) $.
 We can also write $L_i(x, y,z)=L_i'(x, y)+L_i''(z)$ for $i\in \{2, 3\}$ where $L_i': \mathbb F^n_2 \times F^k_2 \rightarrow \mathbb F^k_2$ are affine functions and
 $L_i'':  F^k_2 \rightarrow \mathbb F^k_2$ are linear functions. Set  $F_1''(x)=L_1(x, H (x), I(x))=L_1'(x, H (x))$ and
 $F_2''(x)=(L_2'(x, H(x))+L_2''(I(x)),L_3'(x, H(x))+L_3''(I(x)))$.
 Then, $F''(x)=F_2''\circ F_1''^{-1}(x)$. For any $v\in \mathbb F_2^n$ and $u=(u',u'')\in \mathbb F_2^k \times \mathbb F_2^k $,
 \begin{align*}
 W_{F''}(v,u)=& \sum_{x \in \mathbb F_2^n} (-1)^{<u, F''(x)>+<v, x>}\\
 =& \sum_{x \in \mathbb F_2^n} (-1)^{<u, F''\circ F_1''(x)>+<v, F_1''(x)>}\\
 =&\sum_{x \in \mathbb F_2^n} (-1)^{<u, F''_2(x)>+<v, F_1''(x)>}\\
 =& \sum_{x \in \mathbb F_2^n} (-1)^{<u, (L_2''(I(x)),L_3''(I(x)))>+ <u, (L_2'(x, H(x)),L_3'(x, H(x)))> +<v, L_3'(x, H(x))>}\\
 =& \sum_{x \in \mathbb F_2^n} (-1)^{< L_2''^*(u')+L_3''^*(u''), I(x)>+ <v',H(x)>+<v'',x>+a}.
 \end{align*}
By the construction of $I(x)$, if $L_2''^*(u')+L_3''^*(u'')\neq 0$, $< L_2''^*(u')+L_3''^*(u''), I(x)>+ <v',H(x)>$ is bent. Thus,
$<u, F''(x)>$ is bent when $L_2''^*(u')+L_3''^*(u'')\neq 0$, where $u=(u',u'')\in \mathbb F_2^k \times \mathbb F_2^k $.
For $i=2, 3$, let $A_i$ be the matrices of size $k \times k$ defined as
\begin{align*}
L_i''(z)=z A_i,
\end{align*}
where $z=(z_1, \cdots, z_k) \in \mathbb F_2^k$.
Then,
\begin{align}
L_2''^*(u')+L_3''^*(u'')=&u'A_2^T+u'' A_3^T \nonumber \\
=& (u',u'')  \left [\begin{matrix}
      A_2^T \\
       A_3^T
\end{matrix} \right ].
\end{align}
Recall that $\mathcal L$ is a affine permutation. Hence, the rank of the linear function
$(L_1''(z), L_2''(z), L_3''(z))$ $=(0,L_2''(z), L_3''(z))$ from $\mathbb F_2^k$ to
$\mathbb F^n_2 \times F^k_2 \times F^k_2$  is $k$.
By $(L_1''(z), L_2''(z), L_3''(z))=z \left [ \begin{matrix} 0| A_2 |A_3 \end{matrix} \right]$, the rank of the matrix $\left [ \begin{matrix} A_2 |A_3 \end{matrix} \right]$ is $k$.
Thus, the rank of the matrix $\left [\begin{matrix}
      A_2^T \\
       A_3^T
\end{matrix} \right ]=\left [ \begin{matrix} A_2 |A_3 \end{matrix} \right]^T$ is also $k$.
Set
\begin{align*}
S'=& \{(u',u'')\in \mathbb F_2^k \times \mathbb F_2^k: L_2''^*(u')+L_3''^*(u'')=0\}\\
=&\{(u',u'')\in \mathbb F_2^k \times \mathbb F_2^k: (u',u'')  \left [\begin{matrix}
      A_2^T \\
       A_3^T
\end{matrix} \right ]=0\}.
\end{align*}
Then, $S'$ is a linear
subspace of dimension $k$.
By  the previous discussion, if $u=(u',u'')\in \mathbb F_2^n \setminus S'$,
the component function $<u, F''(x)>$ is bent. Thus, $F''(x)$ has  at least $2^n -2^k$ bent components.
From Theorem 3 in \cite{Pott2017}, $F''(x)$ has exactly $2^n -2^k$ bent components, which completes the proof.

\qed
 \end{proof}

\section{The non-existence of APN plateaued functions having the maximum number of bent components}\label{APNplateaued}
In \cite{Pott2017}, the authors asked if APN functions could have the maximum number of bent components or not. In this section we investigate the case of all APN plateaued functions. The result is given the following theorem.
\begin{theorem}
Let $F$ be a plateaued APN function defined on $\mathbb{F}_2^n$ (where $n\geq 4$ is an even positive integer). Then $F$ cannot have the maximum number of bent components.
\end{theorem}
\begin{proof}
Let $F$ be a plateaued APN function on $\mathbb{F}_2^n$.
Denote 
$$
N_t=\{v\in \mathbb{F}_2^n: W_F(u,v)
=\pm 2^{\frac{n+t}{2}}\}, 
$$
where $W_F(u,v)=\sum_{x\in \mathbb{F}_2^n}
(-1)^{v\cdot F(x)+u\cdot x}$ and $t$ is a positive integer ($0\leq t \leq n$).
We have 
\begin{align}\label{equ1}
\sum_{u,v\in \mathbb{F}_2^n} W_F^4(u,v)
=&\sum_{v\in F_2^n}(2^{\frac{n+t_v}{2}})^2
\sum_{u\in \mathbb{F}_2^n}W_F^2(u,v)\nonumber\\
=&2^n\sum_{v\in \mathbb{F}_2^n}2^{t_v}
\sum_{u\in \mathbb{F}_2^n}W_F^2(u,v)\nonumber\\
=& 2^{3n}\sum_{v\in \mathbb{F}_2^n}2^{t_v}\nonumber\\
=&2^{3n}(N_0+N_22^2+\cdots +N_n2^n).
\end{align}
Since $F$ is APN, we have
\begin{equation}\label{equ2}
\sum_{u,v\in \mathbb{F}_2^n} W_F^4(u,v)
=2^{3n}(3\cdot 2^n-2).
\end{equation}
From Equations (\ref{equ1}) and (\ref{equ2}), we have
$$
N_0+N_22^2+\cdots+N_n 2^n=3\cdot 2^n-2.
$$
Therefore, we have $N_0\equiv 2\mod 4$.
Since $n\geq 4$, $2^n-2^{\frac{n}{2}}
\equiv 0 \mod 4$. Hence,
$$
N_0\neq 2^n-2^{\frac{n}{2}}.
$$
Thus, $F$ does not have  the maximum number of bent components. In particular, quadratic
APN functions cannot have the maximum number of bent components.
\qed
 \end{proof}

\section{New constructions of bent component functions of vectorial functions }\label{newbent}
In this section we provide several functions defined as $G(x) = x  L(x)$
on ${\Bbb F}_{2^{2k}}$ such that the number of bent components $Tr^{2k}_1(\alpha F(x))$ equals $2^{2k}-2^k$, where $L(x)$ is a linear function on  ${\Bbb F}_{2^{2k}}$.
We first recall two lemmas which will be useful in our context.
\begin{lemma}\cite{Pott2017}\label{lemma adjoint}
Let $V = {\Bbb F}_{2^{2k}}$ and let $ <,>$ be a nondegenerate
symmetric bilinear form on $V$. If $\mathcal{L}:V\rightarrow V$ is linear, we
denote the adjoint operator by $\mathcal{L}^*$, i.e., $<x,\mathcal{L}(y)>=<\mathcal{L}^*(x),y>$
for all $x,y\in V$. The function $f : V \rightarrow {\Bbb F}_2$, defined by
$x\mapsto <x,\mathcal{L}(x)>$, is bent if and only if $\mathcal{L}+\mathcal{L}^*$ is invertible.
\end{lemma}

\begin{lemma}\label{lemma L}\cite{Pott2017}
Let $V = {\Bbb F}_{2^{2k}}$ and $<x , y>=Tr^n_1(xy)$ be the trace bilinear form.
 If $\mathcal{L}:V\rightarrow V$ is defined by $\mathcal{L}(x)=\alpha x^{2^i} $, $\alpha\in V$ and for any $i=0,1,\cdots,n-1$, then $\mathcal{L}^*(x)= \alpha^{2^{n-i}} x^{2^{n-i}}$.
\end{lemma}
  In \cite{Pott2017}, the authors presented a construction of bent functions through adjoint operators. We start by providing a  simplified  proof of \cite[Theorem 4]{Pott2017} (which is the main result of their article).
  \begin{theorem}\label{theo in POtt}\cite[Theorem 4]{Pott2017}
  Let $V = {\Bbb F}_{2^{2k}}$ and $i$ be a nonnegative
integer. Then, the mapping $\mathrm{F}_\alpha^{i} $
defined by
\begin{displaymath}
  \mathrm{F}_\alpha^{i}(x)=Tr_{1}^{2k}\left(\alpha x^{2^i}(x+x^{2^k})\right)
  \end{displaymath}
is bent if and only if   $\alpha\notin {\Bbb F}_{2^{k}} $.
  \end{theorem}
  \begin{proof}
  From the proof of \cite[Theorem 4]{Pott2017}, we know
  $$ \mathcal{L}(x)+\mathcal{L}^*(x)=Tr^{2k}_k(\alpha x^{2^i})+\alpha^{2^{2k-i}}\left(Tr^{2k}_k( x)\right)^{2^{k-i}}.$$
  From Lemma \ref{lemma adjoint}, we need to show that
 $\mathcal{L}(x)+\mathcal{L}^*(x)=0$ if and only if $x=0$.

 Let $\nabla_{a}=\{x|Tr^{2k}_k(x)=a, a\in {\Bbb F}_{2^{k}} \} $.  We all know $ Tr^{2k}_k( x) $ is a surjection from ${\Bbb F}_{2^{2k}} $ to ${\Bbb F}_{2^{k}}$ and $\|\nabla_{a}\|=2^{2k-k} $ for any $a\in {\Bbb F}_{2^{k}}$.
 We also know $\nabla_{0}= {\Bbb F}_{2^{k}}$.

  If $\alpha\notin {\Bbb F}_{2^{k}} $, then  $\mathcal{L}(x)+\mathcal{L}^*(x)=0$ if and only if  \begin{equation}\label{equ 1bent2}
  \left\{\begin{array}{c}
Tr^{2k}_k(\alpha x^{2^i})=0,\\
  Tr^{2k}_k( x)=0,
  \end{array} \right.
  \end{equation}
  i.e.,  $x=0$.
  If for any $ x\neq 0$, we always have  $\mathcal{L}(x)+\mathcal{L}^*(x)\neq 0$, then $\alpha\notin {\Bbb F}_{2^{k}} $. In fact, if $\alpha\in {\Bbb F}_{2^{k}} $, then
  $$\mathcal{L}(x)+\mathcal{L}^*(x)=\alpha Tr^{2k}_k( x^{2^i})+\left(\alpha Tr^{2k}_k( x)\right)^{2^{k-i}}=\alpha Tr^{2k}_k( x)+\left(\alpha Tr^{2k}_k( x)\right)^{2^{k-i}}.$$ Further,  $\mathcal{L}(x)+\mathcal{L}^*(x)= 0$ for any $x\in {\Bbb F}_{2^{k}}$.
   Thus,
 we have   $\mathrm{F}_\alpha^{i}(x)$
is bent if and only if   $\alpha\notin {\Bbb F}_{2^{k}}$.
  \qed
  \end{proof}
  Now, we are going to present a first new family of bent functions through adjoint operators.
\begin{theorem}\label{theo-new_bent}
Let $V = {\Bbb F}_{2^{2k}}$ and $i$ be a nonnegative
integer. Let $t_1,t_2$ be two positive integers such that $0\leq t_1,t_2\leq k$ and both $ z^{2^{k-t_1}-1}+z^{2^{k-t_2}-1}+1=0$ and $ z^{2^{t_1}-1}+ z^{2^{t_2}-1}+1=0$ have no solutions on ${\Bbb F}_{2^{k}}$.  Then, the function $\mathrm{F}_\alpha^{i} $ defined on $V$ by
  \begin{equation}\label{equa bent0}
  \mathrm{F}_\alpha^{i}(x)=Tr_1^{2k}\left(\alpha x^{2^i}(x+x^{2^k}+x^{2^{t_1}}+x^{2^{t_1+k}}+x^{2^{t_2}}+x^{2^{t_2+k}})\right)
  \end{equation}
  is bent if and only if $\alpha\notin {\Bbb F}_{2^{k}} $.

\end{theorem}
\begin{proof}
We have
\begin{equation}\label{equa bent1}
\begin{array}{rl}
  \mathrm{F}_\alpha^{i}(x)=&Tr_{1}^{2k}\left(\alpha x^{2^i}(x+x^{2^k}+x^{2^{t_1}}+x^{2^{t_1+k}}+x^{2^{t_2}}+x^{2^{t_2+k}})\right)\\
=&Tr_{1}^{2k}(x\alpha x^{2^i})+Tr_{1}^{2k}(x\alpha^{2^k} x^{2^{i+k}})+Tr_{1}^{2k}(x\alpha^{2^{2k-t_1}} x^{2^{i+2k-t_1}})\\
&+Tr_{1}^{2k}(x\alpha^{2^{k-t_1}} x^{2^{i+k-t_1}})+Tr_{1}^{2k}(x\alpha^{2^{2k-t_2}} x^{2^{i+2k-t_2}})+Tr_{1}^{2k}(x\alpha^{2^{k-t_2}} x^{2^{i+k-t_2}})\\
=&Tr_{1}^{2k}(x\mathcal{L}(x)),
  \end{array}
  \end{equation}
where
\begin{displaymath}\begin{array}{rl}
\mathcal{L}(x)=&\alpha x^{2^i}+\alpha^{2^k} x^{2^{i+k}}+\alpha^{2^{2k-t_1}} x^{2^{i+2k-t_1}}+\alpha^{2^{k-t_1}} x^{2^{i+k-t_1}}\\
&+\alpha^{2^{2k-t_2}} x^{2^{i+2k-t_2}}+\alpha^{2^{k-t_2}} x^{2^{i+k-t_2}}\\
=&\alpha x^{2^i}+(\alpha x^{2^i})^{2^k}+\alpha^{2^{k-t_1}} x^{2^{i+k-t_1}}+(\alpha^{2^{k-t_1}} x^{2^{i+k-t_1}})^{2^k}\\
   &+\alpha^{2^{k-t_2}} x^{2^{i+k-t_2}}+(\alpha^{2^{k-t_2}} x^{2^{i+k-t_2}})^{2^k}\\
  =&\left(\alpha x^{2^i}+(\alpha x^{2^i})^{2^k}\right)+\left(\alpha x^{2^i}+(\alpha x^{2^i})^{2^k}\right)^{2^{k-t_1}}+\left(\alpha x^{2^i}+(\alpha x^{2^i})^{2^k}\right)^{2^{k-t_2}}
\end{array}
\end{displaymath}
According to Lemma \ref{lemma L}, the adjoint operator $ \mathcal{L}^*(x)$ is
\begin{equation}\label{equa adjoint1}
\begin{array}{rl}
 \mathcal{L}^*(x)=&\alpha^{2^{2k-i}} x^{2^{2k-i}}+\alpha^{2^{2k-i}} x^{2^{k-i}}+\alpha^{2^{2k-i}} x^{2^{t_1-i}}\\
  &+\alpha^{2^{2k-i}} x^{2^{k+t_1-i}}+\alpha^{2^{2k-i}} x^{2^{t_2-i}}+\alpha^{2^{2k-i}} x^{2^{k+t_2-i}}\\
  = &\alpha^{2^{2k-i}} \left(x^{2^{2k-i}}+ x^{2^{k-i}}+ x^{2^{t_1-i}}
  + x^{2^{k+t_1-i}}+ x^{2^{t_2-i}}+ x^{2^{k+t_2-i}}\right)\\
  =&\alpha^{2^{2k-i}} \left( x^{2^{k-i}}+(x^{2^{k-i}})^{2^k}+ x^{2^{t_1-i}}
  + (x^{2^{t_1-i}})^{2^k}+ x^{2^{t_2-i}}+ (x^{2^{t_2-i}})^{2^k}\right)\\
  =&\alpha^{2^{2k-i}} \left( (x+x^{2^k})^{2^{k-i}}+ (x+x^{2^k})^{2^{t_1-i}}+(x+x^{2^k})^{2^{t_2-i}}
  \right)\\
  =&\alpha^{2^{2k-i}} \left( (x+x^{2^k})^{2^{k-i}}+ (x+x^{2^k})^{2^{t_1+k-i}}+(x+x^{2^k})^{2^{t_2+k-i}}
  \right)
  \end{array}
  \end{equation}
Thus, we have
\begin{displaymath}
\begin{array}{rl}
   \mathcal{L}(x)+\mathcal{L}^*(x)
  =&\left(\alpha x^{2^i}+(\alpha x^{2^i})^{2^k}\right)+\left(\alpha x^{2^i}+(\alpha x^{2^i})^{2^k}\right)^{2^{k-t_1}}+\left(\alpha x^{2^i}+(\alpha x^{2^i})^{2^k}\right)^{2^{k-t_2}} \\
   &+\alpha^{2^{2k-i}} \left( (x+x^{2^k})^{2^{k-i}}+ (x+x^{2^k})^{2^{t_1+k-i}}+(x+x^{2^k})^{2^{t_2+k-i}}
  \right).
  \end{array}
  \end{displaymath}
  Note that we have $\mathcal{L}(x)\in {\Bbb F}_{2^{k}} $ and $ (x+x^{2^k})^{2^{k-i}}+ (x+x^{2^k})^{2^{t_1+k-i}}+(x+x^{2^k})^{2^{t_2+k-i}}\in {\Bbb F}_{2^{k}}$ for any $x\in {\Bbb F}_{2^{2k}}$.  From Lemma \ref{lemma adjoint}, it is sufficient to
show that $\mathcal{L}(x)+\mathcal{L}^*(x)$ is invertible. That is,  we need to show that
$\mathcal{L}(x)+\mathcal{L}^*(x)=0$ if and only if $x=0$.

 Since both $ z^{2^{k-t_1}-1}+z^{2^{k-t_2}-1}+1=0$ and $ z^{2^{t_1}-1}+ z^{2^{t_2}-1}+1=0$ have no solution in ${\Bbb F}_{2^{k}}$, we have both
   $\left(\alpha x^{2^i}+(\alpha x^{2^i})^{2^k}\right)+\left(\alpha x^{2^i}+(\alpha x^{2^i})^{2^k}\right)^{2^{k-t_1}}+\left(\alpha x^{2^i}+(\alpha x^{2^i})^{2^k}\right)^{2^{k-t_2}}=0$ and $ (x+x^{2^k})^{2^{k-i}}+ (x+x^{2^k})^{2^{t_1+k-i}}+(x+x^{2^k})^{2^{t_2+k-i}}=0$ if and only if
    \begin{equation}\label{equ bent2}
  \left\{\begin{array}{c}
\alpha x^{2^i}+(\alpha x^{2^i})^{2^k}=0,\\
  x+x^{2^k}=0.
  \end{array} \right.
  \end{equation}

  From the proof of Theorem \ref{theo in POtt}, when both $ z^{2^{k-t_1}-1}+z^{2^{k-t_2}-1}+1=0$ and $ z^{2^{t_1}-1}+ z^{2^{t_2}-1}+1=0$ have no solution in ${\Bbb F}_{2^{k}}$,  we have   $\mathrm{F}_\alpha^{i}(x) $ is bent if and only if $\alpha\notin {\Bbb F}_{2^{k}} $.
  \qed
\end{proof}
  We immediately have the following statement by setting $t_2= k-t_1$ in the previous theorem.
\begin{corollary}\label{cor thr}
Let $V = {\Bbb F}_{2^{2k}}$ and $i$ be a nonnegative
integer. Let $t_1,t_2$ be two positive integers such that $t_1+t_2= k$ and $ z^{2^{t_1}-1}+z^{2^{t_2}-1}+1=0$ has no solution in  ${\Bbb F}_{2^{k}}$.  Then, the mapping $\mathrm{F}_\alpha^{i} $ defined on $V$ by
  \begin{displaymath}
  \mathrm{F}_\alpha^{i}(x)=Tr_{1}^{2k}\left(\alpha G(x)\right)
  \end{displaymath}
  is bent if and only if $\alpha\notin {\Bbb F}_{2^{k}} $, where $G(x)=x^{2^i}\left(x+x^{2^k}+x^{2^{t_1}}+x^{2^{t_1+k}}+x^{2^{t_2}}+x^{2^{t_2+k}}\right)$.
\end{corollary}

The previous construction given by Theorem \ref{theo-new_bent} can be  generalized as follows.
\begin{theorem}\label{theorem 3}
Let $V = {\Bbb F}_{2^{2k}}$ and $i$ be a nonnegative
integer. Let $t_1,t_2$ be two positive integers such that $0\leq t_1,t_2\leq k$  and both $ (\gamma^{(1)})^{2^{k-t_1}} z^{2^{k-t_1}-1}+(\gamma^{(2)})^{2^{k-t_2}}z^{2^{k-t_2}-1}+1=0$ and $ (\gamma^{(1)})^{2^{k-i}} z^{2^{t_1}-1}+(\gamma^{(2)})^{2^{k-i}} z^{2^{t_2}-1}+1=0$ have no solution in  ${\Bbb F}_{2^{k}}$, where $\gamma^{(1)}, \gamma^{(2)}\in {\Bbb F}_{2^{k}}$. Then, the mapping $\mathrm{F}_\alpha^{i} $ defined by
  \begin{equation}\label{equa bent 30}
  \mathrm{F}_\alpha^{i}(x)=Tr_{1}^{2k}\left(\alpha x^{2^i}\left(x+x^{2^k}+\gamma^{(1)}(x^{2^{t_1}}+x^{2^{t_1+k}})
  +\gamma^{(2)}(x^{2^{t_2}}+x^{2^{t_2+k}})\right)\right)
  \end{equation}
  is bent if and only if $\alpha\notin {\Bbb F}_{2^{k}} $.

\end{theorem}
\begin{proof}
We have
\begin{equation}\label{equa bent31}
\begin{array}{rl}
  \mathrm{F}_\alpha^{i}(x)
=&Tr_{1}^{2k}(x\mathcal{L}(x)),
  \end{array}
  \end{equation}
where
\begin{displaymath}\begin{array}{rl}
\mathcal{L}(x)=&\alpha x^{2^i}+\alpha^{2^k} x^{2^{i+k}}+(\gamma^{(1)})^{2^{k-t_1}} \left(\alpha^{2^{2k-t_1}} x^{2^{i+2k-t_1}}+\alpha^{2^{k-t_1}} x^{2^{i+k-t_1}}\right)\\
+&(\gamma^{(2)})^{2^{k-t_2}} \left(\alpha^{2^{2k-t_2}} x^{2^{i+2k-t_2}}+\alpha^{2^{k-t_2}} x^{2^{i+k-t_2}}\right)\\
  =&\left(\alpha x^{2^i}+(\alpha x^{2^i})^{2^k}\right)+(\gamma^{(1)})^{2^{k-t_1}} \left(\alpha x^{2^i}+(\alpha x^{2^i})^{2^k}\right)^{2^{k-t_1}}\\
  +&(\gamma^{(2)})^{2^{k-t_2}} \left(\alpha x^{2^i}+(\alpha x^{2^i})^{2^k}\right)^{2^{k-t_2}}
\end{array}
\end{displaymath}
The adjoint operator $ \mathcal{L}^*(x)$ is
\begin{equation}\label{equa adjoint31}
\begin{array}{rl}
 \mathcal{L}^*(x)=&\alpha^{2^{2k-i}} x^{2^{2k-i}}+\alpha^{2^{2k-i}} x^{2^{k-i}}+(\gamma^{(1)})^{2^{k-i}}\left(\alpha^{2^{2k-i}} x^{2^{t_1-i}}\right.\\
  &\left.+\alpha^{2^{2k-i}} x^{2^{k+t_1-i}}\right)+(\gamma^{(2)})^{2^{k-i}}\left(\alpha^{2^{2k-i}} x^{2^{t_2-i}}+\alpha^{2^{2k-i}} x^{2^{k+t_2-i}}\right)\\
  =&\alpha^{2^{2k-i}} \left( (x+x^{2^k})^{2^{k-i}}+ (\gamma^{(1)})^{2^{k-i}}(x+x^{2^k})^{2^{t_1-i}}+(\gamma^{(2)})^{2^{k-i}}(x+x^{2^k})^{2^{t_2-i}}
  \right)\\
  =&\alpha^{2^{2k-i}} \left( (x+x^{2^k})^{2^{k-i}}+ (\gamma^{(1)})^{2^{k-i}}(x+x^{2^k})^{2^{t_1+k-i}}
  +(\gamma^{(2)})^{2^{k-i}}(x+x^{2^k})^{2^{t_2+k-i}}
  \right)
  \end{array}
  \end{equation}
  Note that we have $\mathcal{L}(x)\in {\Bbb F}_{2^{k}} $ and $ (x+x^{2^k})^{2^{k-i}}+ (\gamma^{(1)})^{2^{k-i}}(x+x^{2^k})^{2^{t_1+k-i}}
  +(\gamma^{(2)})^{2^{k-i}}(x+x^{2^k})^{2^{t_2+k-i}}\in {\Bbb F}_{2^{k}}$ for any $x\in {\Bbb F}_{2^{2k}}$.
  In order to
show that $\mathcal{L}(x)+\mathcal{L}^*(x)$ is invertible, we need to show that
$\mathcal{L}(x)+\mathcal{L}^*(x)=0$ if and only if $x=0$.

 Since  both $ (\gamma^{(1)})^{2^{k-t_1}} z^{2^{k-t_1}-1}+(\gamma^{(2)})^{2^{k-t_2}}z^{2^{k-t_2}-1}+1=0$ and $ (\gamma^{(1)})^{2^{k-i}} z^{2^{t_1}-1}+(\gamma^{(2)})^{2^{k-i}} z^{2^{t_2}-1}+1=0$ have no solutions on ${\Bbb F}_{2^{k}}$, we have both $\mathcal{L}(x)=0$ and $  (x+x^{2^k})^{2^{k-i}}+ (\gamma^{(1)})^{2^{k-i}}(x+x^{2^k})^{2^{t_1+k-i}}
  +(\gamma^{(2)})^{2^{k-i}}(x+x^{2^k})^{2^{t_2+k-i}}=0$ if and only if
  \begin{equation}\label{equ bent32}
  \left\{\begin{array}{c}
\alpha x^{2^i}+(\alpha x^{2^i})^{2^k}=0,\\
  x+x^{2^k}=0.
  \end{array} \right.
  \end{equation}

 By using the proof of Theorem \ref{theo in POtt}, we have   $\mathrm{F}_\alpha^{i}(x) $ is bent if and only if $\alpha\notin {\Bbb F}_{2^{k}} $. \qed


\end{proof}
By the same process used to prove Theorem \ref{theorem 3}, one can get the following result.
\begin{theorem}\label{theo bentkla}
Let  $V = {\Bbb F}_{2^{2k}}$ and
$i,\rho$ be two nonnegative
integers such that $\rho\leq k$. Let $\gamma^{(j)}\in {\Bbb F}_{2^{k}}$ and $0\leq t_j\leq k$ be a nonnegative
integer, where $j=1,2,\cdots,\rho$. Assume that both equations  $ \sum\limits_{j=1}^{\rho}(\gamma^{(j)})^{2^{k-t_j}} z^{2^{k-t_j}-1}+1=0$ and $ \sum\limits_{j=1}^{\rho}(\gamma^{(j)})^{2^{k-i}} z^{2^{t_j}-1}+1=0$ have no solution in  ${\Bbb F}_{2^{k}}$.  Then, the mapping $\mathrm{F}_\alpha^{i} $ defined on $V$ by
  \begin{equation}\label{equa bent 41}
  \mathrm{F}_\alpha^{i}(x)=Tr_{1}^{2k}\left(\alpha G(x)\right)
  \end{equation}
  is bent if and only if $\alpha\notin {\Bbb F}_{2^{k}} $, where
  $G(x)= x^{2^i}\left(Tr^{2k}_k(x)+\sum\limits_{j=1}^{\rho}\gamma^{(j)}(Tr^{2k}_k(x))^{2^{t_j}}\right)$.

\end{theorem}


\begin{lemma}\cite{Pott2017}\label{lemma vectorial bent}
 Let $F_\alpha(x) = Tr^{2k}_1(\alpha G(x))$, be a
Boolean bent function for any $\alpha \in {\Bbb F}_2^{2k}\setminus {\Bbb F}_2^{k}$, where
$G : {\Bbb F}_2^{2k}\rightarrow {\Bbb F}_2^{2k}$. Then, $F : {\Bbb F}_2^{2k}\rightarrow {\Bbb F}_2^{k} $,
defined as $F(x)=Tr^{2k}_k(\alpha G(x))$ is a vectorial bent function
for any $\alpha \in {\Bbb F}_2^{2k}\setminus {\Bbb F}_2^{k}$.
\end{lemma}
  According to Theorem \ref{theo bentkla} and Lemma \ref{lemma vectorial bent}, we immediately get the following theorem.
  \begin{theorem}\label{theo bentvector}
Let  $G(x)$ be defined as in Theorem \ref{theo bentkla}.
 Then, the mapping $\mathrm{F}_\alpha$ defined by
  \begin{equation}\label{equa bent 41}
  \mathrm{F}_\alpha(x)=Tr^{2k}_k\left(\alpha G(x)\right)
  \end{equation}
  is a vectorial bent function
for any $\alpha \in {\Bbb F}_{2^{2k}}\setminus {\Bbb F}_{2^k}$.
\end{theorem}

  In \cite{Pott2017}, the authors presented the differential spectrum of the
functions $ G: {\Bbb F}_{2^{2k}}\rightarrow {\Bbb F}_{2^{2k}}$ defined by $G(x)=x^{2^i}(x+x^{2^k})$. Their result is given below.
\begin{lemma}\cite{Pott2017}\label{lemma diff}
   Let $i$ be a nonnegative integer such that $i<k$. The differential spectrum of the
functions   $G(x)=x^{2^i}(x+x^{2^k}), G: {\Bbb F}_{2^{2k}}\rightarrow {\Bbb F}_{2^{2k}} $, is given by,
   \begin{equation}\label{equ lem pott}
   \delta_{G}(a,b)\in \left\{
   \begin{array}{cl}
  \{0,2^k\} & ~~\text{if}~~a\in {\Bbb F}^*_{2^k},\\
  \{0,2^{\gcd(i,k)}\}& ~~\text{if}~~a\in {\Bbb F}_{2^{2k}}\setminus {\Bbb F}_{2^k}.
   \end{array}\right.
   \end{equation}
In particular, $\delta_{G}(a,b)=2^k$ only for $a\in {\Bbb F}^*_{2^k}$ and $b\in {\Bbb F}_{2^k}$.

\end{lemma}
Now we are going to show that the differential  spectrum of the
functions $x^{2^i}(x+x^{2^k}+x^{2^{t_1}}+x^{2^{t_1+k}}+x^{2^{t_2}}+x^{2^{t_2+k}}) $ is different from the one of the functions $x\mapsto x^{2^i}(x+x^{2^k}) $.
{
\begin{theorem}
 Let  $ \mathrm{F}_\alpha^{i}(x)=Tr^{2k}_1\left(\alpha G(x)\right)$
  be defined as Theorem \ref{theo-new_bent}, where $G(x)= x^{2^i}(x+x^{2^k}+x^{2^{t_1}}+x^{2^{t_1+k}}+x^{2^{t_2}}+x^{2^{t_2+k}})$.
   If 
   there exists
   $t_1=1$ and
 $gcd(t_2,k)\neq 1$ such that both $ z^{2^{k-t_1}-1}+z^{2^{k-t_2}-1}+1=0$ and $ z^{2^{t_1}-1}+ z^{2^{t_2}-1}+1=0$ have no solutions on ${\Bbb F}_{2^{k}}$, then for $i=t_2$, there exist elements $a\in {\Bbb F}_{2^{2k}}\setminus {\Bbb F}_{2^k}$ such that  the number
  $\delta_G(a,b)$ is equal to $2$ for any $b\in {\Bbb F}_{2^{2k}}$, which is neither   $2^{\gcd(i,k)}$  nor $ 2^k$.
\end{theorem}
\begin{proof}
Let $a\in {\Bbb F}_{2^{2k}}\setminus {\Bbb F}_{2^k}$ such that  $\tau^{2^{t_1}}=\tau$, where  $\tau=a+a^{2^k}\neq 0$ (since $t_1|k$).
We have
\begin{equation}\label{equ bent51}
 \begin{array}{rl}
G(x)+G(x+a)=&a^{2^i}\left(x+x^{2^k}+(x+x^{2^k})^{2^{t_1}}+(x+x^{2^k})^{2^{t_2}}\right)\\
&+(x+a)^{2^i}\left(a+a^{2^k}+(a+a^{2^k})^{2^{t_1}}+(a+a^{2^k})^{2^{t_2}}\right).
  \end{array}
  \end{equation}
   Thus, for any $x'\in {\Bbb F}_{2^{2k}}$, there must be one element $b\in {\Bbb F}_{2^{2k}}$ such that $ G(x')+G(x'+a)=b$.


  Let $x',x''$ are the solutions of $G(x)+G(x+a)=b $. Hence
    \begin{equation}\label{equ bent52}
 \begin{array}{rl}
&G(x')+G(x'+a)+G(x'')+G(x''+a)\\
=&a^{2^i}\left(x'+x'^{2^k}+(x'+x'^{2^k})^{2^{t_1}}+(x'+x'^{2^k})^{2^{t_2}}
+x''+x''^{2^k}\right.\\
&\left.+(x''+x''^{2^k})^{2^{t_1}}+(x''+x''^{2^k})^{2^{t_2}}\right)\\
&+(x'+x'')^{2^i}\left(a+a^{2^k}+(a+a^{2^k})^{2^{t_1}}+(a+a^{2^k})^{2^{t_2}}\right)=0.
  \end{array}
  \end{equation}
  Since $x+x^{2^k}\in {\Bbb F}_{2^k}$ for any $x\in {\Bbb F}_{2^{2k}}$, (\ref{equ bent52}) implies that $(x'+x'')^{2^i}\left(a+a^{2^k}+(a+a^{2^k})^{2^{t_1}}\right.$ $\left.+(a+a^{2^k})^{2^{t_2}}\right)$ belongs to the
 multiplicative coset $a^{2^i}{\Bbb F}^*_{2^k}$.  Thus, we necessarily have $x'+x''=a\nu$,
where $\nu\in {\Bbb F}^*_{2^k}$.
Further, $x'+x''+(x'+x'')^{2^k}=a\nu+a^{2^k}\nu$. Since $t_1|t_2$, from (\ref{equ bent52}), we have
\begin{equation}\label{equ bent53}
\begin{array}{rl}
&\left(\tau+\tau^{2^{t_1}}+\tau^{2^{t_2}}\right) \nu^{2^i}
+\tau \nu+\tau^{2^{t_1}}\nu^{2^{t_1}}+\tau^{2^{t_2}}\nu^{2^{t_2}}\\
=&\tau \left(\nu^{2^{i}}
+ \nu+ \nu^{2^{t_1}}+ \nu^{2^{t_2}}\right)=0.
\end{array}
\end{equation}
If we set $i=t_2$, then from (\ref{equ bent53}) we have $\delta_G(a,b)=2\neq 2^{\gcd(i,k)}$ since $\gcd(i,k)=gcd(t_2,k)\neq 1$.

For one element $b\in {\Bbb F}_{2^{2k}}$, if for any $x'\in {\Bbb F}_{2^{2k}}$,  we always have $ G(x')+G(x'+a)\neq b$, then $\delta_G(a,b)=0$.

  \qed
\end{proof}

}

\begin{theorem}
Let 
$i,\rho$ be two nonnegative
integers such that $\rho\leq k$. Let $0\leq t_j\leq k$ be a nonnegative
integer, where $j=1,2,\cdots,\rho$. Assume that both $ \sum\limits_{j=1}^{\rho}z^{2^{k-t_j}-1}+1=0$ and $ \sum\limits_{j=1}^{\rho} z^{2^{t_j}-1}+1=0$ have no solution in  ${\Bbb F}_{2^{k}}$.  Then, the mapping $\mathrm{F}_\alpha^{i} $ defined by
  \begin{equation}
  \mathrm{F}_\alpha^{i}(x)=Tr_{1}^{2k}\left(\alpha G(x)\right)
  \end{equation}
  where
  $G(x)= x^{2^i}\left(Tr^{2k}_k(x)+\sum\limits_{j=1}^{\rho}(Tr^{2k}_k(x))^{2^{t_j}}\right)$
  is bent if and only if $\alpha\notin {\Bbb F}_{2^{k}} $. 
 Further, if the number of the solutions of $ \sum\limits_{j=1}^{\rho} z^{2^{t_j}}+z+z^{2^i}=0$  on ${\Bbb F}_{2^{k}}$ is not equal to $2^{\gcd(i,k)}$, then there exist elements $a\in {\Bbb F}_{2^{2k}}\setminus {\Bbb F}_{2^k}$ such that  the number
  $\delta_G(a,b)$ does not equal  $2^{\gcd(i,k)}$ for any $b\in {\Bbb F}_{2^{2k}}$.

\end{theorem}
\begin{proof}
From Theorem \ref{theo bentkla}, we know $ \mathrm{F}_\alpha^{i}(x) $  is bent if and only if $\alpha\notin {\Bbb F}_{2^{k}} $.
We have
\begin{equation}\label{equ bent61}
 \begin{array}{rl}
G(x)+G(x+a)=&a^{2^i}\left(Tr^{2k}_k(x)+\sum\limits_{j=1}^{\rho}(Tr^{2k}_k(x))^{2^{t_j}}\right)\\
&+(x+a)^{2^i}\left(Tr^{2k}_k(a)+\sum\limits_{j=1}^{\rho}(Tr^{2k}_k(a))^{2^{t_j}}\right)=b.
  \end{array}
  \end{equation}

     Let $a\in {\Bbb F}_{2^{2k}}\setminus {\Bbb F}_{2^k}$ such that    $a+a^{2^k}=1 $.  We need to show the number of  solutions of $ G(x)+G(x+a)=b$ is not equal to $ 2^{\gcd(i,k)}$ for any
$b\in {\Bbb F}_{2^{2k}}$.  Let $\rho=\gcd(i,k)$.
    We suppose  $\delta_G(a,b)=2^{\rho}$ and let $x',x''$ are the solutions of (\ref{equ bent61}) for some $b$. Hence
    \begin{equation}\label{equ bent62}
 \begin{array}{rl}
&G(x')+G(x'+a)+G(x'')+G(x''+a)\\
=&a^{2^i}\left(Tr^{2k}_k(x'+x'')+\sum\limits_{j=1}^{\rho}(Tr^{2k}_k(x'+x''))^{2^{t_j}}\right)+(x'+x'')^{2^i}=0
  \end{array}
  \end{equation}
  since $ \sum\limits_{j=1}^{\rho} z^{2^{t_j}-1}+1=0$ have no solution in  ${\Bbb F}_{2^{k}} $, that is, $ \left(Tr^{2k}_k(a)+\sum\limits_{j=1}^{\rho}(Tr^{2k}_k(a))^{2^{t_j}}\right)=1$.
  For any $x\in {\Bbb F}_{2^{2k}}$, (\ref{equ bent62}) implies that $(x'+x'')^{2^i}$ belongs to the
 multiplicative coset $a^{2^i}{\Bbb F}^*_{2^k}$.  Thus, we necessarily have $x'+x''=a\nu$,
where $\nu\in {\Bbb F}^*_{2^k}$.
Further, $Tr^{2k}_k(x'+x'')=x'+x''+(x'+x'')^{2^k}=a\nu+a^{2^k}\nu$.  From (\ref{equ bent62}), we have
\begin{equation}\label{equ bent63}
 \nu^{2^i}
+\nu+\sum\limits_{j=1}^{\rho}\nu^{2^{t_j}}=0.
\end{equation}
We also know that the number of the solutions of $ \sum\limits_{j=1}^{\rho} z^{2^{t_j}}+z+z^{2^i}=0$  on ${\Bbb F}_{2^{k}}$ is not equal to $2^{\gcd(i,k)}$, thus, if  $a\in \{ x|Tr^{2k}_k(x)=1,x\in {\Bbb F}_{2^{2k}}\}$,  the number
  $\delta_G(a,b)$ is not equals  $2^{\gcd(i,k)}$ for any $b\in {\Bbb F}_{2^{2k}}$
  \qed
\end{proof}

\begin{theorem}\label{theo genealized}
Let $n=2k,e$ be two positive integers. Let $V = {\Bbb F}_{2^{2k}}$ and $i$ be a nonnegative
integer. Let $E=\{x|x\in {\Bbb F}_{2^{2k}}, Tr_k^{2k}(x)\in {\Bbb F}_{2^{e}} \} $ and
 $O=\{x\in {\Bbb F}_{2^{2k}}, Tr_k^{2k}(x)\in M \} $, where  $M=\{y+Tr^k_e(y)| y\in{\Bbb F}_{2^{k}}\}$.
 Let $\mathrm{F}_\alpha^{i} $ be the function defined on $V$ by
  \begin{equation}\label{equa bent70}
  \mathrm{F}_\alpha^{i}(x)=Tr_{1}^{2k}\left(\alpha x^{2^i}Tr^{2k}_e(x)\right).
  \end{equation}
  If $\frac{k}{e}$ is even, then $\mathrm{F}_\alpha^{i} $
  is bent if and only if $\alpha\notin E$.
    If $\frac{k}{e}$ is odd, then $\mathrm{F}_\alpha^{i} $
  is bent if and only if $\alpha\notin O$.
   Further,
  if $k$ is odd and $e=2$, then $\mathrm{F}_\alpha^{i} $
  is bent if and only if $\alpha\notin O$.
\end{theorem}
  \begin{proof}
We have
\begin{equation}\label{equa bent71}
\begin{array}{rl}
  \mathrm{F}_\alpha^{i}(x)=&Tr_{1}^{2k}\left(\alpha x^{2^i}(x+x^{2^e}+x^{2^{2e}}+\cdots+x^{2^{2k-e}})\right)\\
=&Tr_{1}^{2k}(x\alpha x^{2^i})+Tr_{1}^{2k}(x\alpha^{2^{2k-e}} x^{2^{i+2k-e}})+Tr_{1}^{2k}(x\alpha^{2^{2k-2e}} x^{2^{i+2k-2e}})\\
&+\cdots +Tr_{1}^{2k}(x\alpha^{2^{e}} x^{2^{i+e}})\\
=&Tr_{1}^{2k}(x\mathcal{L}(x)),
  \end{array}
  \end{equation}
where
\begin{displaymath}\begin{array}{rl}
\mathcal{L}(x)=&\alpha x^{2^i}+\alpha^{2^{2k-e}} x^{2^{i+2k-e}}+\alpha^{2^{2k-2e}} x^{2^{i+2k-2e}}
+\cdots +\alpha^{2^{e}} x^{2^{i+e}}\\
=&\alpha x^{2^i}+(\alpha x^{2^i})^{2^{2k-e}}+(\alpha x^{2^i})^{2^{2k-2e}}+\cdots+(\alpha x^{2^i})^{2^{e}}\\
=&Tr^{2k}_e(\alpha x^{2^i}).

\end{array}
\end{displaymath}
According to Lemma \ref{lemma L}, the adjoint operator $ \mathcal{L}^*(x)$ is
\begin{equation}\label{equa adjoint71}
\begin{array}{rl}
 \mathcal{L}^*(x)=&\alpha^{2^{2k-i}} x^{2^{2k-i}}+\alpha^{2^{2k-i}} x^{2^{e-i}}+\alpha^{2^{2k-i}} x^{2^{2e-i}}
  +\cdots+ \alpha^{2^{2k-i}} x^{2^{2k-e-i}}\\
  = &\alpha^{2^{2k-i}} \left(x^{2^{2k-i}}+ x^{2^{e-i}}+ x^{2^{2e-i}}
   +\cdots+  x^{2^{2k-e-i}}\right)\\
   = &\alpha^{2^{2k-i}} \left(x+ x^{2^{e}}+ x^{2^{2e}}
   +\cdots+  x^{2^{2k-e}}\right)^{2^{2k-i}}\\
   \\
=&\alpha^{2^{2k-i}}\left(Tr^{2k}_e( x)\right)^{2^{2k-i}}.
  \end{array}
  \end{equation}
Thus, we have
\begin{displaymath}
\begin{array}{rl}
   \mathcal{L}(x)+\mathcal{L}^*(x)
  =&Tr^{2k}_e(\alpha x^{2^i})+\alpha^{2^{2k-i}}\left(Tr^{2k}_e( x)\right)^{2^{2k-i}}.
  \end{array}
  \end{displaymath}
  From Lemma \ref{lemma adjoint}, it is sufficient to
show that $\mathcal{L}(x)+\mathcal{L}^*(x)$ is invertible. That is,  we need to show that
$\mathcal{L}(x)+\mathcal{L}^*(x)=0$ if and only if $x=0$.

For $\frac{k}{e}$ being even,   we have $ Tr^{2k}_e( x)=0$ if and only if $x\in E$.
If $\alpha\notin E$, then $\mathcal{L}(x)+\mathcal{L}^*(x)=0$ is only if
  \begin{equation}\label{equ bent72}
  \left\{\begin{array}{c}
Tr^{2k}_e(\alpha x^{2^i})=Tr^{k}_e\left(Tr^{2k}_k(\alpha x^{2^i})\right)=0,\\
 Tr^{2k}_e( x)=0,
  \end{array} \right.
  \end{equation}
  i.e., $x=0$.  If for any $x\neq 0$,  we have $\mathcal{L}(x)+\mathcal{L}^*(x)\neq0 $, then $\alpha\notin E$. In fact, if suppose $\alpha\in E $, then
  \begin{displaymath}
  \begin{array}{rl}
   \mathcal{L}(x)+\mathcal{L}^*(x)
  =&Tr^{2k}_e(\alpha x^{2^i})+\alpha^{2^{2k-i}}\left(Tr^{2k}_e( x)\right)^{2^{2k-i}}\\
  =&
  Tr^{2k}_k(\alpha )Tr^{k}_e\left(Tr^{2k}_k( x^{2^i})\right)+\alpha^{2^{2k-i}}\left(Tr^{2k}_e( x)\right)^{2^{2k-i}}\\
  =&0
   \end{array}
   \end{displaymath}
   for any $x\in E$. Hence, $  \mathrm{F}_\alpha^{i}(x)$
  is bent if and only if $\alpha\notin E$.

  Similarly, for $\frac{k}{e}$  odd,  we have $ Tr^{2k}_e( x)=0$ if and only if $x\in O$.  We can prove  $  \mathrm{F}_\alpha^{i}(x)$
  is bent if and only if $\alpha\notin O$. 

   Similarly, for $k$  odd and $e=2$,  we have $ Tr^{2k}_2( x)=0$ if and only if $x\in O$.  We can prove  $  \mathrm{F}_\alpha^{i}(x)$
  is bent if and only if $\alpha\notin O$.
  \qed
\end{proof}

\begin{remark}
Note that Theorem \ref{theo in POtt} is special case of Theorem \ref{theo genealized}. It corresponds to the case where $e=k$.
\end{remark}
Similary to Theorem \ref{theo bentkla}, we have the following statement.
\begin{theorem}\label{theo bentkla generalized}
Let 
$i,\rho$ be two nonnegative
integers such that $\rho\leq k$. Let $\gamma^{(j)}\in {\Bbb F}_{2^{k}}$ and $0\leq t_j\leq k$ be a nonnegative
integer, where $j=1,2,\cdots,\rho$. Let  both $ \sum\limits_{j=1}^{\rho}(\gamma^{(j)})^{2^{k-t_j}} z^{2^{k-t_j}-1}+1=0$ and $ \sum\limits_{j=1}^{\rho}(\gamma^{(j)})^{2^{k-i}} z^{2^{t_j}-1}+1=0$ have no solution in ${\Bbb F}_{2^{k}}$. Let $E$ and
 $O$ be defined as Theorem \ref{theo genealized}.  Let
 the function $\mathrm{F}_\alpha^{i} $ be defined by
  \begin{equation}\label{equa bent 80}
  \mathrm{F}_\alpha^{i}(x)=Tr_{1}^{2k}\left(\alpha x^{2^i}\left(Tr^{2k}_e(x)
  +\sum\limits_{j=1}^{\rho}\gamma^{(j)}(Tr^{2k}_e(x))^{2^{t_j}}\right)\right).
  \end{equation}
  If $\frac{k}{e}$ is even, then $\mathrm{F}_\alpha^{i} $
  is bent if and only if $\alpha\notin E$.
    If $\frac{k}{e}$ is odd, then $\mathrm{F}_\alpha^{i} $
  is bent if and only if $\alpha\notin O$.
   Further,
  if $k$ is odd and $e=2$, then $\mathrm{F}_\alpha^{i} $
  is bent if and only if $\alpha\notin O$.
\end{theorem}

\section{Conclusions}
This paper is in the line of a very recent paper published in the IEEE-transactions Information Theory by Pott et al  \cite{Pott2017} in which several open problems have been raised.
In the present paper, we have established that the property of a function having the maximal number of bent components is invariant under CCZ-equivalence which gives an answer to an open problem in \cite{Pott2017}. Next, we have proved the non-existence of APN plateaued functions having  the maximal number of bent components which gives a partial answer to an open problem in \cite{Pott2017}.
Furthermore, we have exhibited several bent functions $F_{\alpha}^i$  for any $\alpha\in\mathbb{F}_{2^{2k}}\setminus\mathbb{F}_{2^k}$ provided that some conditions hold. In other words, the set of those $\alpha$ for which $F_{\alpha}^i$ is bent is of maximal cardinality $2^{2k}-2^{k}$. This provide an answer to another open problem in \cite{Pott2017}. In addition, we have studied the differential spectrum of certain functions and showed that it is not equal to those studied in \cite{Pott2017}.

\end{document}